\documentclass[pre,citeautoscript,superscriptaddress,twocolumn,nofootinbib,floatfix]{revtex4-1}

\pdfoutput=1

\usepackage[utf8]{inputenc}
\usepackage[T1]{fontenc}

\def \gentitle {Significant Scales in Community Structure}

\usepackage{amsmath,amsthm,amssymb}
\usepackage{xcolor}
\usepackage{graphicx}
\definecolor{link-color}{cmyk}{0.8 ,  0.3 ,  0. , 0}
\usepackage[colorlinks=true,%
            linkcolor   = {link-color},%
            citecolor   = {link-color},%
            urlcolor    = {link-color}]{hyperref}
\usepackage[caption=false]{subfig}

\newtheorem{theorem}{Theorem}

\newcommand{\Hf}{\mathcal{H}}

\newcommand{\res}{\gamma}
\newcommand{\samecomm}{\delta(\sigma_i, \sigma_j)}

\newcommand{\E}{\mathbb{E}}

\newcommand{\Sig}{\mathcal{S}}
\newcommand{\G}{\mathcal{G}}
\newcommand{\midd}{\parallel}
\newcommand{\eqtxt}{equation}
\DeclareMathOperator{\argmax}{arg\,max}

\graphicspath{ {../figures/} }

\begin{document}

  \author{V.A. Traag}
  \email{vincent.traag@uclouvain.be}
  \affiliation{ICTEAM, Universit\'e catholique de Louvain}
  \affiliation{Royal Netherlands Institute of Southeast Asian and Caribbean Studies}
  \author{G. Krings}
  \affiliation{Real Impact Analytics}
  \author{P. Van Dooren}
  \affiliation{ICTEAM, Universit\'e catholique de Louvain}

  \title{\gentitle}

\begin{abstract}

  Many complex networks show signs of modular structure, uncovered by community
  detection. Although many methods succeed in revealing various partitions, it
  remains difficult to detect at what scale some partition is significant. This
  problem shows foremost in multi-resolution methods. We here introduce an
  efficient method for scanning for resolutions in one such method.
  Additionally, we introduce the notion of ``significance'' of a partition,
  based on subgraph probabilities. Significance is independent of the exact
  method used, so could also be applied in other methods, and can be interpreted
  as the gain in encoding a graph by making use of a partition. Using
  significance, we can determine ``good'' resolution parameters, which we
  demonstrate on benchmark networks. Moreover, optimizing significance itself
  also shows excellent performance. We demonstrate our method on voting data
  from the European Parliament. Our analysis suggests the European Parliament
  has become increasingly ideologically divided and that nationality plays no
  role.

\end{abstract}

\maketitle

\section{Introduction}

Networks appear naturally in many fields of science, and are often inherently
complex structures. By looking at the modular structure of a network we can
reduce its complexity to some extent, yielding a ``bird's-eye view'' of the
network~\cite{Fortunato2010,Porter2009,Newman2012}. 

Although there is no universally accepted definition of a community, there are
some commonly accepted principles. We denote by $G = (V,E)$ a graph with nodes
$V$ and edges $E \subseteq V \times V$, where the graph has $n=|V|$ number of
nodes and $m=|E|$ number of edges, and is said to have a density of $p =
m/\binom{n}{2}$. The idea is that in general, we want to reward links within
communities with some weight $a_{ij}$, while we want to punish missing links
within communities with some weight $b_{ij}$. Working out this idea we arrive at
\begin{equation}
  \Hf(\sigma) = -\sum_{ij} \bigl[a_{ij} A_{ij} - b_{ij} (1 - A_{ij}) \bigr]\samecomm,
\end{equation}
for the ``cost'' of a partition $\sigma$. Here $A_{ij}$ is the adjacency
matrix, which is $A_{ij} = 1$ if there is a link between $i$ and $j$ and zero
otherwise, $\sigma_i$ denotes the community of node $i$, and $\samecomm = 1$ if
and only if $\sigma_i = \sigma_j$ and zero otherwise. This is a slightly more
simplified version of the approach by Reichardt and
Bornholdt~\cite{Reichardt2006Statistical}. We will restrict ourselves here to
simple, unweighted graphs.

Different weights $a_{ij}$ and $b_{ij}$ give rise to different methods. One can
imagine for example taking the number of common neighbours as weight $b_{ij}$,
the distance of the shortest path or some transition probability in a random
walk. Many methods have been developed over the years, but the most noteworthy
method is that of modularity~\cite{Newman2004Finding} which uses $a_{ij} = 1 -
p_{ij}$, $b_{ij} = p_{ij}$ where $p_{ij}$ is some random null-model. It has
risen to prominence because it showed encouraging results in various fields,
ranging from ecology~\cite{Guimera2010,Stouffer2011} and
biology~\cite{Guimera2005Functional,Kashtan2005} to political
science~\cite{Mucha2010} and sociology~\cite{Conover2013}. 

Nonetheless modularity was found to be seriously flawed. Its biggest problem is
the resolution limit~\cite{Fortunato2007a,Kumpula2007a}, which states that
modularity is unable to detect relatively small communities in large networks.
We showed previously that methods that use local weights (i.e.  $a_{ij}$ and
$b_{ij}$ are independent of the graph) do not suffer from the resolution
limit~\cite{Traag2011}, and are hence called resolution limit free. Within this
framework there are relatively few methods that are resolution limit free. One
such method is the Constant Potts Model~\cite{Traag2011} (CPM). This model has
as weights $a_{ij} = 1 - \res$ and $b_{ij} = \res$ where $\res$ is a so-called
resolution parameter (see next paragraph), resulting in 
\begin{equation}
  \Hf(\sigma,\res) = -\sum_{ij} \bigl[ A_{ij} - \res \bigr] \samecomm.
  \label{equ:CPM}
\end{equation}
Rewriting this in terms of communities, we arrive at
\begin{equation}
  \Hf(\sigma,\res) = -\sum_c \bigl[e_{c} - \res n_c^2\bigr],
  \label{equ:CPM_per_comm}
\end{equation}
where $e_c$ is the number of edges within community $c$ (or twice\footnote{This
is due to double counting in $\sum_{ij} (A_{ij} -
\res)\delta(\sigma_i, \sigma_j)$ for undirected graphs} for undirected graphs)
and $n_c$ is the number of nodes within community $c$. It can be seen as a
variant of the Reichardt and Bornholdt Potts model when choosing an
Erd\"os-R\'enyi (ER) null model, which assumes that each edge has the same
independent probability of being included $p$. In the remainder of this article,
when speaking of a random graph, we refer to an ER random graph, unless
explicitly stated otherwise.

It is not too difficult to show that any (local) minimum yields a nice
interpretation of the role of the resolution parameter $\res$. Succinctly
stated, communities have an internal density of at least $\res$ and an external
density of at most $\res$. The parameter $\res$ can thus be seen as the desired
density of the communities. The central question in this paper is why and how we
should choose some resolution parameter $\res$.

\section{Results}

Although CPM does not suffer from the resolution limit, there do remain some
problems of scale~\cite{Lancichinetti2011}. In particular, there is no a-priori
way to choose a particular resolution parameter $\res$. We address this issue in
this paper from two complementary perspectives. First we will detail how to
efficiently scan different resolution parameters $\res$ for CPM in
section~\ref{sec:scan_bisect}.  Secondly, we introduce the notion of
``significance'' of a partition (which is independent of any method) in
section~\ref{sec:significance}. Both perspectives help in choosing some
particular resolution parameter $\res$. We will demonstrate the method on
benchmark networks, and show that both scanning for the right resolution
parameter as well as optimizing significance itself shows excellent performance.
As an application of our method, we analyse a network based on votes of the
European Parliament (EP).

\subsection{Scanning resolutions}
\label{sec:scan_bisect}

Often, various measures of stability---how much does the partition change after
some perturbation---are used to determine whether a resolution parameter or a
partition is
``good''~\cite{Mirshahvalad2012,Ronhovde2009,Delvenne2010,Mirshahvalad2013,Mirshahvalad2012}.
In this section we look at stable ``plateaus'': ranges of $\res$ where the same
partition is optimal. If a partition is optimal over the
range of $[\res_1,\res_2]$ then the communities have a density of at least
$\res_2$ and are separated by a density of at most $\res_1$. Hence, the larger
this stable ``plateau'', the more clear-cut the community structure.

For $\res = 0$, the trivial partition of all nodes in a single community is
optimal (since in that case any cut will increase the cost function). On the
other hand, for $\res = 1$ the optimal partition is to have each node in its own
community. This idea holds in general: a higher $\res$ gives rise to smaller
communities. 

The intuitive idea that a partition should remain optimal for some (continuous)
interval of $\res$ can be formalized. More precisely, if $\sigma$ is an optimal
solution for $\res_1$ and $\res_2$, then $\sigma$ is also an optimal solution
for all $\res \in [\res_1, \res_2]$ (which was also remarked in the supporting
information of ref.~\cite{Mucha2010} for a similar method).

\begin{theorem}
Let $\Hf(\res,\sigma)$ be as in \eqtxt~\eqref{equ:CPM}. If $\sigma^*$ is optimal
for both $\res_1$ and $\res_2$, or
$$\sigma^* = \argmax_\sigma \Hf(\res_1,\sigma) = \argmax_\sigma
\Hf(\res_2,\sigma)$$ 
then $\sigma^* = \argmax_\sigma \Hf(\res,\sigma)$ for $\res_1 \leq
\res \leq \res_2$.
\end{theorem}
\begin{proof}
  First observe that $\Hf(\res, \sigma)$ is linear in $\res$, which can be
  easily seen from the definition. Suppose that $\sigma^*$ is optimal in
  $\res_1$ and $\res_2$. Let $\res = \lambda \res_1 + ( 1 - \lambda) \res_2$
  with $0 \leq \lambda \leq 1$, then by linearity of $\Hf(\res, \sigma)$ in
  $\res$ and optimality of $\sigma^*$ we have
  \begin{align*}
    \Hf(\res, \sigma^*) 
      &= \lambda \Hf(\res_1, \sigma^*) + (1 - \lambda) \Hf(\res_2, \sigma^*), \\
      &\leq \lambda \Hf(\res_1, \sigma) + (1 - \lambda) \Hf(\res_2, \sigma)
      =\Hf(\res, \sigma).
  \end{align*}
  Hence $\Hf(\res, \sigma^*) \leq \Hf(\res, \sigma)$ and $\sigma^*$ is
  optimal for $\res \in [\res_1, \res_2]$.
\end{proof}
As stated, $\Hf(\res, \sigma)$ is linear in $\res$, and we can rewrite it
slightly to emphasize its linearity 
\begin{align}
  \Hf(\res, \sigma) &= - \sum_{ij} \bigl[A_{ij} - \res\bigr] \samecomm \nonumber \\
                    &= -\bigl[E - \res N\bigr]
\end{align}
where $E := \sum_{c} e_c$ the total of internal edges and $N := \sum_{c} n_c^2$
is the sum of the squared community sizes. 

It is less obvious how to detect whether a partition remains optimal over some
interval. Fortunately, it turns out that $N$ is monotonically decreasing with
$\res$. Specifically, if both partitions are optimal for both resolution
parameters, then necessarily $N_1 = N_2$, and so also $E_1 = E_2$. We therefore
only need to find those points at which $N(\res)$ changes, which can be done
efficiently using bisectioning on $\res$.

\begin{theorem}
Let $\sigma_z = \argmax_\sigma \Hf(\res_z,\sigma)$, $z=1,2$. Furthermore,
let $N_z = \sum_c n^2_c(\sigma_z)$ where $n_c(\sigma_z)$ denote the
community sizes of the partition $\sigma_z$. If $\res_1 < \res_2$
then $N_1 \geq N_2$.
\end{theorem}
\begin{proof}
The two partitions $\sigma_1$ and $\sigma_2$ have the costs $\Hf(\res_1,
\sigma_1) = -E_1 + \res_1 N_1$, $\Hf(\res_2, \sigma_2) = -E_2 + \res_2 N_2$.
Both partitions are optimal for the corresponding resolution
parameters and we obtain
\begin{align*}
  -E_1 + \res_1 N_1 &\leq -E_2 + \res_1 N_2, \\
  -E_2 + \res_2 N_2 &\leq -E_1 + \res_2 N_1.
\end{align*}
Summing both inequalities yields
\begin{equation*}
  -(E_1 + E_2) + \res_1 N_1 + \res_2 N_2 \leq -(E_1 + E_2) + \res_1 N_2 + \res_2 N_1
\end{equation*}
and so $\res_1 (N_1 - N_2) \leq \res_2 (N_1 - N_2)$. Since $\res_1 < \res_2$
we obtain that $N_1 \geq N_2$. 
\end{proof}

\subsection{Significance}
\label{sec:significance}

\begin{figure}
\begin{center}
  \includegraphics{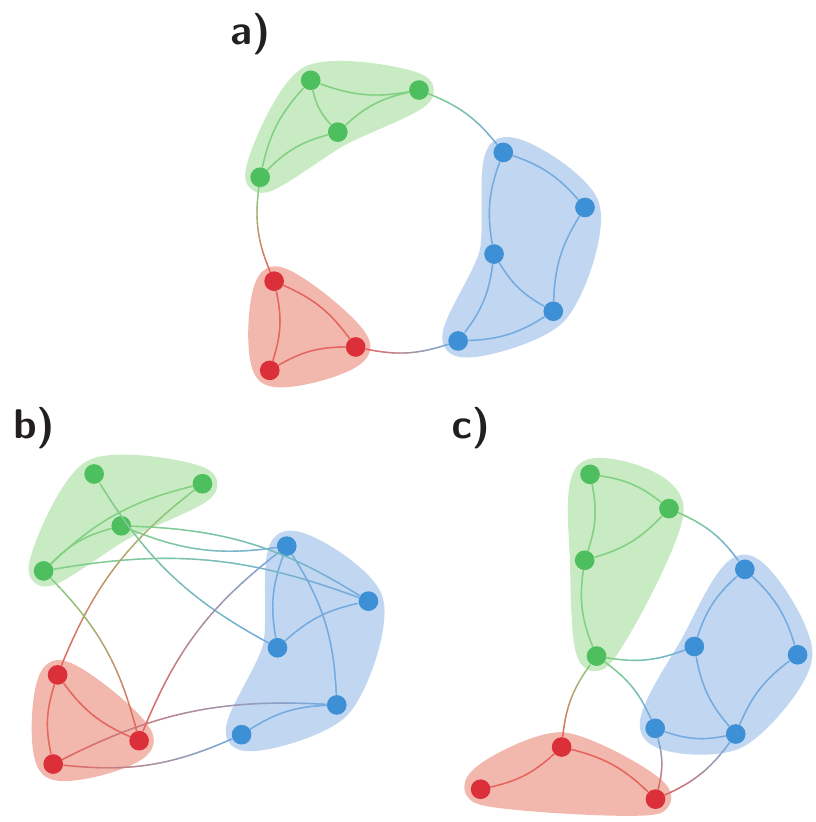}
\end{center}
\caption{Probabilities for partitions. Consider the example partition provided
  in (a). The objective is to somehow estimate how (un)likely such a partition
  occurs in a random graph---the significance of a partition. In (b) and (c) we
  show the same graph, but in (b) the same partition as in (a) is used, while in
  (c) a partition with more internal edges is used. For illustrative purposes,
  the graph is generated by randomly rewiring some of the edges and permuting
  the nodes of the original graph in (a). Earlier approaches keep the partition
  fixed, and focus on the probability that so many edges fall within the given
  partition, as illustrated in (b). Yet this ignores there might exist some
  partition within this graph that has more internal edges. Therefore, we focus
  on the probability of \emph{finding} such a dense partition in random graphs,
as illustrated in (c).}
  \label{fig:sig_diff}
\end{figure}

Another, complementary, point of view would be to have some quality measure to
state at what resolution $\res$ the partition is ``good''. After some
reflection, it is ironic we return to the question of what resolution yields a
good partition.  After all, the initial goal of modularity was in fact to decide
on some resolution level: where to cut a particular
dendrogram~\cite{Newman2004Finding}.

Although modularity compares the number of edges within a community to a random
graph, this does not provide any ``significance'' of a partition, since random
graphs and sparse graphs without community structure can also have quite high
modularity~\cite{Reichardt2006,Bagrow2012a,Montgolfier2011}. Other approaches
have been suggested that try to estimate in some way the significance of a
partition. One recent approach, known as ``surprise'', focuses on the
probability to find $E$ internal edges in a random
graph~\cite{Aldecoa2011,Aldecoa2013}. Another more ``local'' approach keeps the
degrees constant and asks what the probability is to connect so many edges to a
given community~\cite{Lancichinetti2010}, which led to a method known as
OSLOM~\cite{Lancichinetti2011a}. A third approach focuses on the likelihood of
generating a graph given a certain partition and degree
distribution~\cite{Karrer2011}, known as stochastic block models.

But when thinking about the significance of a partition, most methods go about
it the wrong way
around~\cite{Lancichinetti2010,Lancichinetti2011a,Karrer2011,Aldecoa2011,Aldecoa2013}.
We do \emph{not} want to know the probability a ``fixed'' partition contains at
least $E$ internal edges, but whether a partition with at least $E$ internal
edges \emph{can be found} in a random graph, which is the approach we will take
in this paper. After all, community detection involves searching for some good
partition, so we should focus on the probability of \emph{finding} such a good
partition in a random graph. In a way, the earlier approaches assume the
partition is ``fixed'' and the edges are randomly distributed, whereas we try to
find a partition in a random graph, which can result in quite different
statistics. Stated somewhat differently, earlier approaches ignore that a simple
permutation of nodes still contains the same partition---one only needs to
identify the permutation to uncover the original partition---whereas our
approach does account for that. We illustrate the differences in the two
approaches in Fig.~\ref{fig:sig_diff}

Nonetheless, these earlier approaches might work quite well. For example,
explicitly calculating the probability to find $E$ internal edges, seems to
yield good results~\cite{Aldecoa2011,Aldecoa2013}. Obviously, the two
probabilities---surprise and our approach---are not completely independent. If
the probability of finding many edges within a partition is high then surely
finding a partition with many edges should be easy. On the other hand, if the
probability of finding a dense partition is low, then surely the probability a
partition contains many edges is low as well. In between these two extremes is a
grey area, and a more in-depth analysis is required for understanding it
exactly. 

Although exact results for finding a partition in a random graph are hard to
obtain, we do get some interesting asymptotic results. The asymptotic limit we
analyse concerns the probability to find a partition into a fixed number of
communities with a certain density for $n \to \infty$ in a random graph. The
probability for finding a certain partition can be reduced to finding some dense
subgraphs in a random graph. We consider subgraphs of size proportional to $n$,
so that it is of size $s n$, with $0 < s < 1$ of a fixed density $q$. Our
central result concerning these subgraph probabilities is the following (the
proof can be found in section~\ref{sec:subgraph_prob_proof}). We here use the
asymptotic notation $f = \Theta(g)$ for denoting $g$ is an asymptotic upper and
lower bound for $f$.

\begin{theorem}
The probability that a subgraph of size $n_c$ and density $q$ appears in a
random graph of size $n$ and density $p$ is asymptotically
\begin{equation}
  \Pr(S(n_c, q) \subseteq \G(n,p)) = e^{\Theta\left(-\binom{n_c}{2} D(q \midd p)\right)}
\end{equation}
where $D(q \midd p)$ is the Kullback-Leibler divergence~\cite{Cover2012}
\begin{equation}
  D(q \midd p) = q \log \frac{q}{p} + (1 - q) \log \frac{1 - q}{1 -p}.
  \label{equ:KL}
\end{equation}
\end{theorem}

For each $p\neq q$ the probability decays as a Gaussian, with a rate depending
on the ``distance'' between $p$ and $q$ as expressed by the Kullback-Leibler
divergence. Furthermore, the larger the subgraph the less likely a subgraph of
different density than $p$ can be found. Combining these probabilities we arrive
at the following approximation for the probability for a partition to be
contained in a random graph
\begin{equation}
\Pr(\sigma) = \prod_c \exp \left( - \binom{n_c}{2} D(p_c \midd p) \right)
\label{equ:prob_partition}
\end{equation}
where $p_c$ is the density of community $c$. We define the significance then
as\footnote{For directed graphs, it is more appropriate to use $\sum_c n_c(n_c
- 1) D(p_c \midd p)$.}
\begin{equation}
\Sig(\sigma)  = - \log \Pr(\sigma) = \sum_c \binom{n_c}{2} D(p_c \midd p).
\end{equation}
Notice that for the two trivial partitions of (1) all nodes in a single
community ($\res = 0$) or (2) each node in its own community ($\res = 1$), the
significance is zero (assuming no self-loops). Since the significance is
non-negative (because the Kullback-Leibler divergence is non-negative), there
will most likely be some partition in between these two extremes ($0 < \res <
1$) which yields a non-zero significance.

\subsubsection{Encoding gain}

Notice that the Kullback-Leibler divergence can be interpreted as a kind of
entropy difference. It can be written as
\begin{equation}
D(q \midd p) = H(q,p) - H(q)
\end{equation}
where $H(q)$ is the binary entropy and $H(q,p)$ is the cross entropy
\begin{align}
H(q) &= -q \log q - (1-q) \log (1-q), \\
H(q,p) &= -q \log p - (1-q) \log (1-p).
\end{align}
Hence, it measures the difference in entropy between $p$ and $q$, assuming that
$q$ is the ``correct'' probability. 

This points to a possible interpretation of the significance $\Sig(\sigma)$ in
terms of encoding of the graph. Suppose we are requested to compress the graph
$G$, and we do so using the simplest possible framework: for each possible edge
we indicate whether it is present or not. Using the average graph density $p$,
by Shannon's source coding theorem~\cite{Cover2012}, the optimal code lengths
are $-\log p$ for indicating an edge is present and $-\log (1-p)$ for indicating
an edge is absent. Now suppose that for some community we have the actual
density $q$. The expected code length using the average graph density is then
$H(q,p)$. If we use the actual graph density $q$ however, we obtain an expected
code length of $H(q)$. The gain in coding efficiency by using $q$ instead of $p$
is then $D(q \midd p)$. Doing so for all $\binom{n_c}{2}$ possible edges, and
for all communities then yields the significance (we hence don't count the
external edges). Significance can thus be regarded as the gain in encoding a
graph by making use of a partition.

\subsubsection{Using significance}

There are two ways to use significance. Firstly, we could use significance to
select a particular resolution parameter $\res$. As was made clear in
section~\ref{sec:scan_bisect}, we don't have to scan $\res \in [\res_1, \res_2]$
if $N(\res_1) = N(\res_2)$. If in addition we are only interested in the $\res$
for which $\Sig(\sigma)$ is maximal, we can only scan those ranges for which the
significance is maximal (taking a greedy approach), similar to root-finding
bisectioning.

Secondly, we could optimise significance itself. We use an approach similar to
the Louvain method~\cite{Blondel2008} for optimizing significance (see
section~\ref{sec:optimize_sig}). Notice that using significance as an objective
function is not resolution limit free, contrary to CPM~\cite{Traag2011}. After
all, given a partition and a graph, pick a subgraph that consists of only a
single community. Then the significance $\Sig(\sigma)$ of that partition,
defined on the subgraph equals $0$, since $D(p_c \midd p) = 0$. Since this
constitutes the minimum, it is unlikely that no other partition provides a
higher significance. Hence, the same partition no longer (necessarily) remains
optimal on all community induced subgraphs, and the method is hence not
resolution limit free.

\subsection{Resolution Profile}

\begin{figure*}
\begin{center}
  \includegraphics{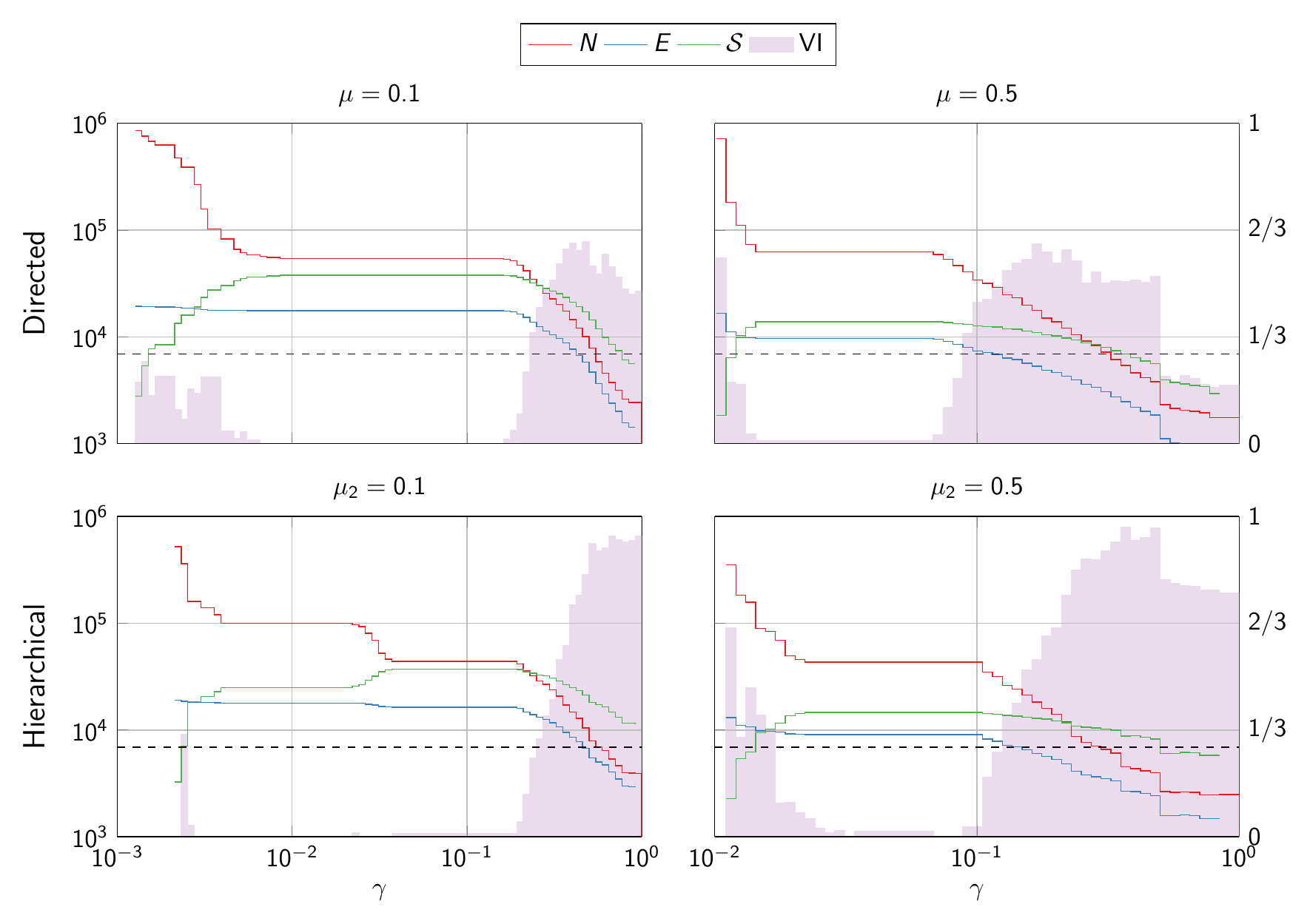}
\end{center}
\caption{Scanning results for directed and hierarchical benchmark graphs. We
display the squared community sizes $N = \sum_c n_c^2$, the total internal
edges $E = \sum_c e_c$ and the significance $\mathcal{S} = \sum_c
\binom{n_c}{2} D(p_c \midd p)$ of each partition (all on a logarithmic axis on
the left). The VI (on a linear axis on the right) is calculated over the various
results returned by running a stochastic algorithm. If the VI is low, this
indicates that the partitions found by the algorithm are (almost) the same. The
black dashed line indicates the expected (maximal) significance of an equivalent
random graph, which is estimated to be about $n \log n \approx 6\,908$.}
\label{fig:scan_res_sig}
\end{figure*}

Scanning the resolution parameters using bisectioning seems to work quite well
on LFR benchmark networks~\cite{Lancichinetti2008a}, as displayed in
Fig.~\ref{fig:scan_res_sig}. These benchmark networks have $n=10^3$ nodes and
have an average degree $\langle k \rangle = 20$ with a maximum degree of $\Delta
= 50$, and follow a power-law distribution $k^{\tau_k}$ with $\tau_k = 2$. The
community sizes range between $20$ and $100$, and are distributed according to
$n_c^{\tau_c}$ with $\tau_c = 1$. This corresponds to the settings as used for
comparing several algorithms~\cite{Lancichinetti2009}. The proportion of
internal links can be controlled by a so-called mixing parameter $0 \leq \mu
\leq 1$, so that for $\mu = 0$ communities are easily detectable, whereas this
becomes increasingly difficult for higher $\mu$. For the hierarchical benchmark
the mixing parameters $\mu_1$ controls the coarser level and $\mu_2$ controls
the finer level. For more details, we refer to Lancichinetti, Fortunato \&
Radicchi~\cite{Lancichinetti2008a}.

\begin{figure*}
\begin{center}
  \includegraphics{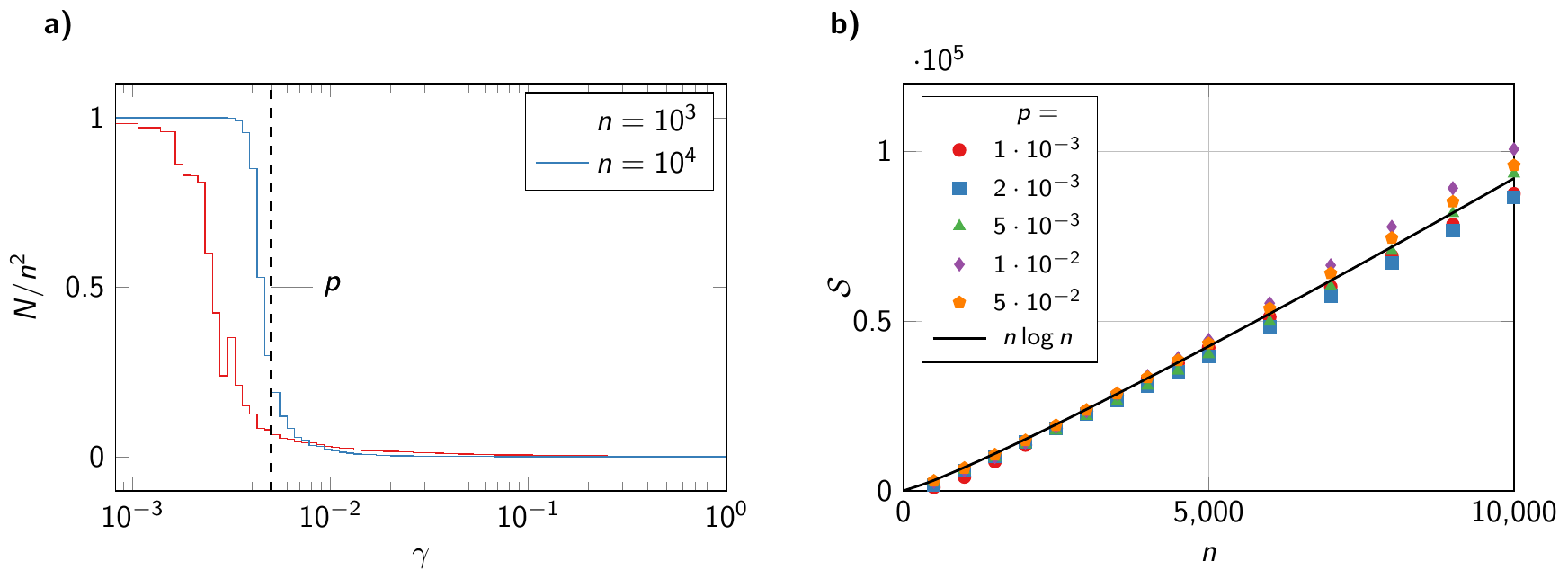}
\end{center}
\caption{Results for ER graphs. In (a) we show that there is a transition around
  $\res = p$ the density of the graph. This transition can be explained by the
  subgraph probabilities calculated in this paper, which suggest that
  asymptotically, a random graph only contains subgraphs  of about the same
density (of size proportional to $n$). In (b) we show the significance of random
graphs, which seems to scale approximately with $n \log n$.}
\label{fig:ER}
\end{figure*}

From Fig.~\ref{fig:scan_res_sig} it is quite clear that both $N$ and $E$ are
stepwise decreasing functions of $\res$. The plateaus indeed correspond to the
planted partition for the benchmark network. The ``stability'' of a partition is
reported in terms of the average pairwise variation of information (VI) between
the various results of multiple runs of the algorithm. The VI measure can be
interpreted as a distance between partitions~\cite{Meila2007}, so a low value
indicates the results are relatively stable. Indeed, in the range of the
plateau, the VI is relatively low (near $0$), indicating the partition is
relatively stable. Hence, using such heuristics, it seems possible to scan for
``stable'' plateaus of resolution values. Moreover, significance is highest in
the region of the plateaus, and thus seems to be able to point to ``meaningful''
resolutions for these networks.

For hierarchical LFR benchmark graphs~\cite{Lancichinetti2008a} results are
similar (Fig.~\ref{fig:scan_res_sig}). This network has $n=10^3$ nodes, and each
node has a degree of $k_i = k = 20$. It consists of $10$ large communities of
$100$ nodes each, and each large community is composed of $5$ smaller
communities of $20$ nodes each. We observe two plateaus for $\mu_2 = 0.1$ (we
have used $\mu_1 = 0.1$ for both results), corresponding to the two levels of
the hierarchy. For these plateaus the VI is near zero, indicating quite stable
results. For $\mu_2 = 0.5$ the two plateaus have merged into a single plateau,
The smaller communities are more significant for $\mu_2 = 0.1$. This makes
sense, since the smaller communities are quite well defined for this regime,
while the larger communities are less clearly defined. Interestingly, when the
two plateaus merge for $\mu_2 = 0.5$, the significance is lower than for $\mu_2
= 0.1$. Indeed, the communities are less clearly defined for $\mu_2 = 0.5$ than
for $\mu_2 = 0.1$. Again, this makes sense, as the smaller communities are much
less clearly defined, while most links still fall within the larger community
(since $\mu_1 = 0.1$).

\subsection{ER graphs}
\label{sec:ER}

Applying the same technique as in the previous subsection to ER graphs, we
obtain a resolution profile, which shows a particular transition
(Fig.~\ref{fig:ER}a). This transition can be explained by the asymptotics of
significance. As the graph grows, and $n \to \infty$, the probability in
\eqtxt~\eqref{equ:prob_partition} $\Pr(\sigma) \to 0$ for $p_c \neq p$. This
indicates that it becomes increasingly difficult to find (relatively large)
subgraphs of a density different from $p$, and in the limit we expect only to
find subgraphs of about density $p$. For $\res < p$ we then expect to find one
large community, while for $\res > p$ we expect to obtain each node in its own
community, thereby explaining the transition around $\res^* \approx p$. The
asymptotic analysis ignores the fact that the number of communities may grow
with the number of nodes. Therefore, it misses the fact that small communities
may have a density of $p_c > p$, which explains the somewhat slower increase of
number of communities for $\res > p$.

Analysing how significance behaves in ER graphs provides us with a baseline to
compare to observed significance values. Obviously, the maximum significance
scales with the size of the graph. In particular, it seems to scale as $n \log
n$ (Fig.~\ref{fig:ER}b). Compared to the benchmark graphs
(Fig.~\ref{fig:scan_res_sig}), the significance found in random graphs is
rather low, so that significance shows little to no sign of any community
structure in ER graphs (although there will be a non-trivial partition obtaining
this maximum significance). By comparing the observed significance in any graph
to $n \log n$, one is thus able to asses to what extent the observed community
structure is significant. We believe this represents a first step towards a
fully fledged hypothesis testing of the significance of community structure.

\subsection{Maximizing significance benchmarks}

\begin{figure*}
\begin{center}
  \includegraphics{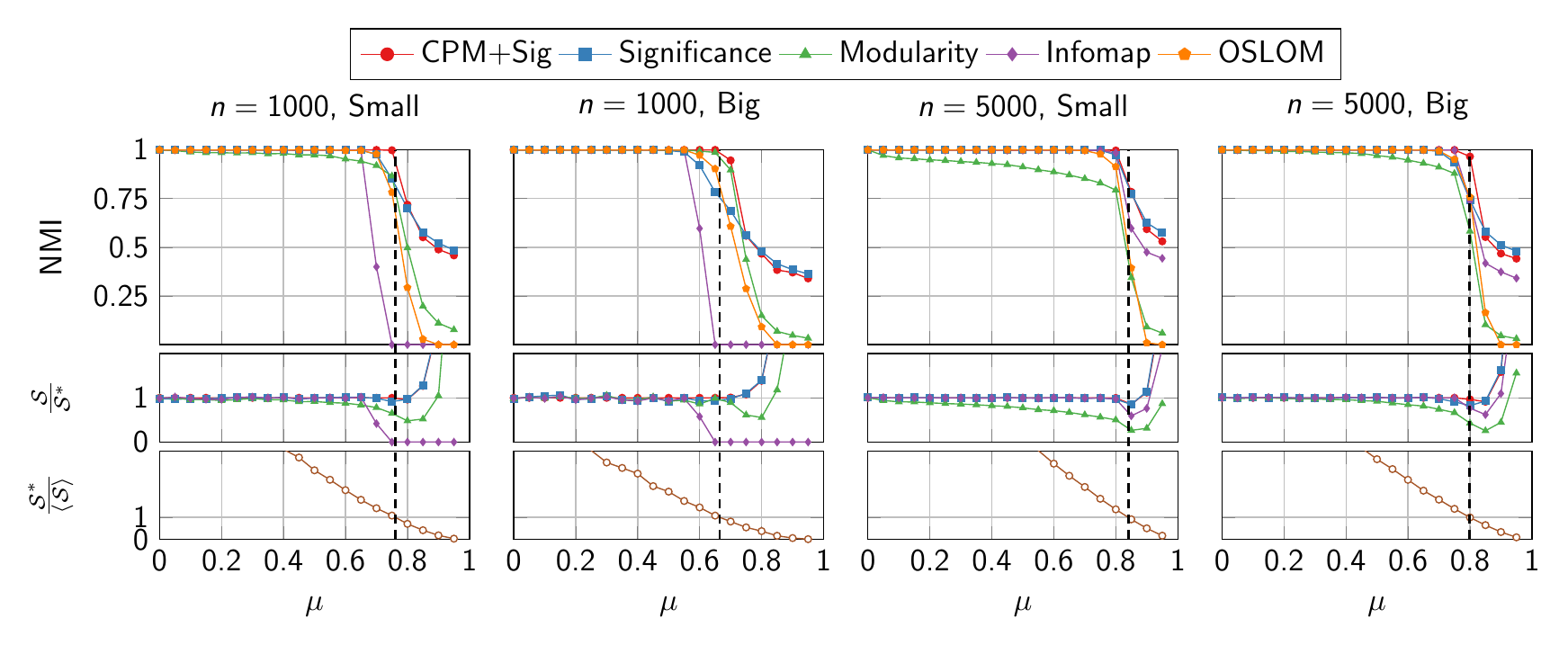}
\end{center}
\caption{Benchmark results for significance. Finding the optimal resolution
  value for CPM using significance seems to work best (first row), where an NMI
  of $1$ indicates the algorithm uncovers exactly the planted partition (OSLOM
  might return overlapping communities, we used an adjusted NMI for
  that~\cite{Lancichinetti2009b}, but it is still $1$ if it is correct).
  Optimizing significance itself also works rather well. We tested two different
  community size distributions: the small communities are between $10$ and $50$
  nodes, and the big communities between $20$ and $100$ nodes. The resolution
  limit is clearly visible for modularity, which shows especially for small
  groups in large networks. Whenever the significance found by each method
  $\mathcal{S}$ is higher than the significance of the planted partition
  $\mathcal{S}^*$, the planted partition is no longer the optimal partition from
  the significance point of view (second row). If the significance of the
  planted partition $\mathcal{S}^*$ is lower than the significance of an
  equivalent random graph $\langle \mathcal{S} \rangle$ no methods seems able to
  correctly detect the planted partition (third row).}
\label{fig:benchmark}
\end{figure*}

We have tested two methods: (1) using significance to choose a $\res$ in
CPM; and (2) optimizing significance itself. We used the standard LFR benchmark,
with the same parameters as for Fig.~\ref{fig:scan_res_sig} for the ``big''
communities, while the ``small'' communities range from 10  to 50, for both
$n=1000$ and $n=5000$. The results are displayed in Fig.~\ref{fig:benchmark}.
We measure the performance using the normalized mutual information
(NMI)~\cite{Lancichinetti2009}, with $\mathrm{NMI} = 1$ indicating the method
uncovered the planted partition exactly. It is clear that using significance to
scan for the best $\res$ parameter for CPM works quite well. Surprisingly
however, optimizing significance itself results in a slightly worse performance
than scanning for the optimal $\res$ parameter for CPM for some settings. This
is presumably due to some local minima in which the significance optimization
gets stuck, while this is not the case for CPM.  Nonetheless, optimizing
significance works quite well, and seems to outperform
Infomap~\cite{Rosvall2011,Rosvall2008}, which was previously shown to perform
well~\cite{Lancichinetti2009}. The OSLOM method performs relatively well,
although not as well as using significance to scan for the best $\res$ parameter
for CPM. This method is aimed at overlapping communities, so an adjusted
NMI~\cite{Lancichinetti2009b} was used to account for that, which still equals
$1$ if it uncovers the planted partition exactly. No results for the
significance are provided, since there is no adjusted version of this measure
(yet).  Modularity clearly shows signs of the resolution
limit~\cite{Fortunato2007a}, as it has difficulties detecting smaller
communities in relatively large networks.  In general, all methods have a
similar computational complexity and use (variants of) the Louvain
method~\cite{Blondel2008}. Detecting the optimal resolution value $\res$ for CPM
involves running the Louvain method multiple times which obviously takes more
time.

Calculating the significance for the planted partition $\mathcal{S}^*$, we see
that in general whenever a method correctly finds the planted communities (i.e.
$\mathrm{NMI} = 1$), that the significance of the partition found by that
algorithm is equivalent, so that $\mathcal{S}=\mathcal{S}^*$ (second row of
Fig.~\ref{fig:benchmark}). We observe a decrease in significance for increasing
$\mu$, as expected (third row of Fig.~\ref{fig:benchmark}). At the point where
the significance of the planted partition goes below the significance of an
equivalent random graph, $\mathcal{S}^* < \langle \mathcal{S} \rangle$, no
method seems able to correctly detect the communities. This suggests that
significance accurately captures whether there is some partition present in the
network or not. Before this point, whenever a method is unable to detect the
planted communities, the significance of that ``incorrect'' partition is lower
than that of the planted partition, $\mathcal{S} < \mathcal{S}^*$, indicating
that the planted partition is of maximal significance.

\subsection{European Parliament}

\begin{figure*}
  \begin{center}
    \includegraphics{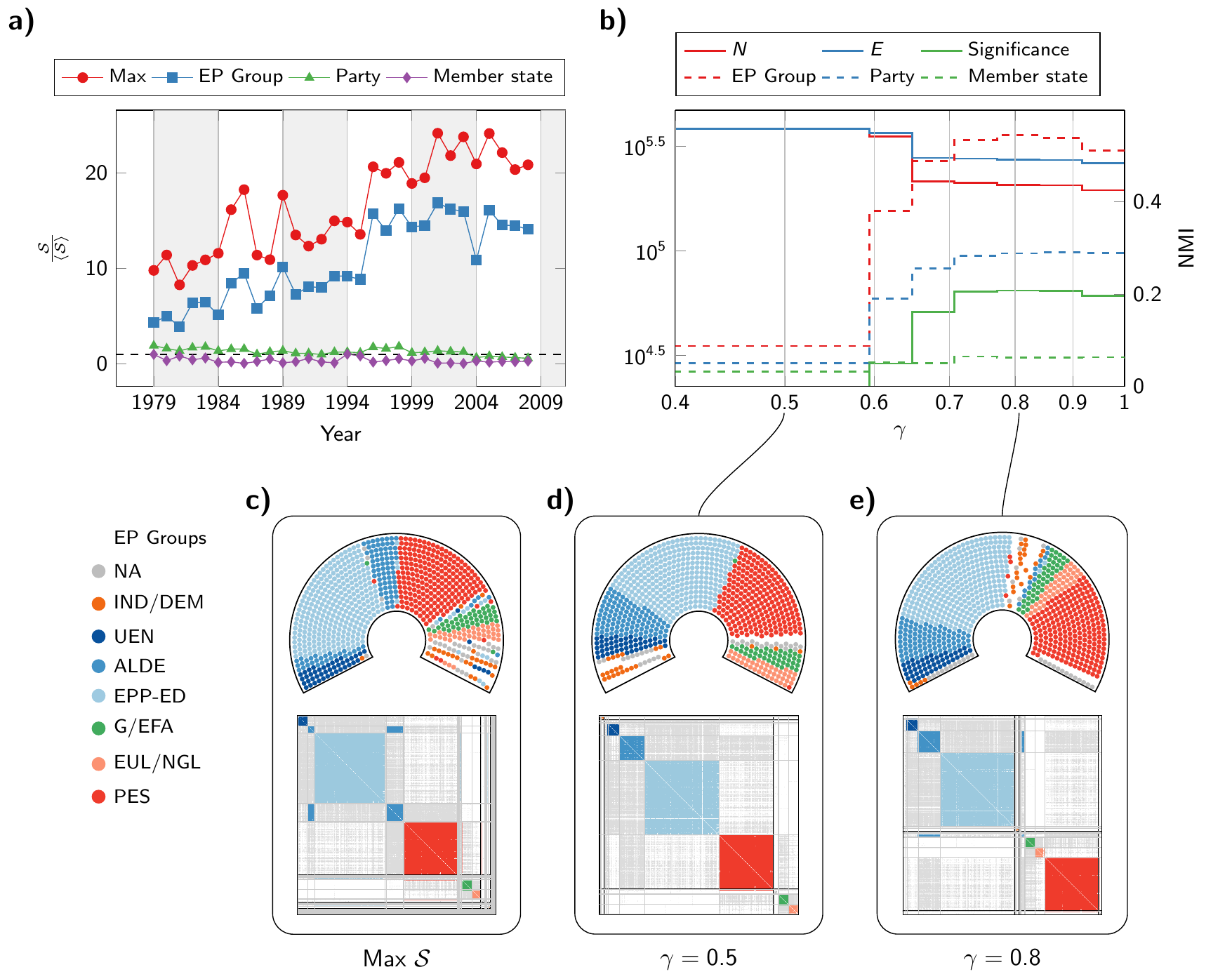}
  \end{center}
  \caption{Results for the European Parliament (EP). In (a) we show the
    significance of four different possible partitions throughout time: the
    partition that maximizes significance and partitions based on the
    affiliation of each member of parliament (MEP) to an EP group, a national
    party or a member state. In (b) we show the resolution profile for the sixth
    EP in the parliamentary year 2008 (June 13, 2008---June 12, 2009).  Besides
    the other quantities, we also show the similarity as measured by the NMI to
    partitions based on the EP groups, the national parties and the member
    states. In (c), (d) and (e) we show how such a partition looks like at
    maximum significance and at two different resolution values $\res=0.5$ and
    $\res=0.8$ respectively. The latter corresponds to the partition that
    maximizes significance. The top shows the division in parliament, while the
    bottom shows the adjacency matrix ordered the same as the parliament. The
    division in communities is indicated by the grouping of the seats in
    parliament and the black lines in the adjacency matrix, while the EP groups
    are indicated by colour. For a key to the abbreviations of the EP parties,
    we refer to the main text.}
  \label{fig:MEP}
\end{figure*}

We demonstrate the method on networks of the European Parliament (EP) from
1979--2009, where each vote of a member of parliament (MEP) for or against a
certain proposal is recorded, the so-called roll call votes (these do not
constitute all votes in the EP though), similar to an analysis of the U.S.
Senate~\cite{Mucha2010}. Over this whole period, a total of almost $16$ million
votes were cast, by in total a little over $2\,500$ different MEPs for more than
$21\,000$ issues. For each parliamentary year (roughly from mid-June to mid-June
the next year), we constructed a network, where there is a link between two MEPs
whenever they vote more in accord than average. We only take into account votes
whenever both MEPs cast a yea or nay vote (instead of abstaining, not voting or
being absent). We used data from
\href{http://personal.lse.ac.uk/hix/HixNouryRolandEPdata.htm}{Simon
Hix}~\cite{Hix2007}.

The MEPs are elected for a five year period from national member states, and
each MEP is associated to a national party. In total we can discern 169 national
parties over the whole period, but usually parties and MEPs organise themselves
in political groups (EP groups) that correspond to some ideological views,
ranging from liberalism to socialism and from conservatives to progressives. Not
all MEPs organise themselves in EP groups; these are known as Non-Attached (NA)
members. Although the EP has the power to choose the European Commission (not
per individual commissioner, but as a whole), they do not need to organise
themselves in governing parties and opposition. Nonetheless, various coalitions
are formed, and from time to time the largest groups have collaborated in a
grand coalition of sorts~\cite{Kreppel2003}. In short, we can create a partition
on three different aspects of the MEPs: (1) their EP group; (2) their national
party; and (3) their member state. In addition, we obtain the partition that
maximizes significance.

We show the normalized significance (i.e. normalized by $\langle S \rangle
\approx n \log n$) for the four different possible partitions in
Fig.~\ref{fig:MEP} from 1979 (the first EP) to 2008 (the sixth EP). Given the
(sub) national constituencies of elected MEPs, one particular concern is that
the EP is governed by national interests, rather than some common European
interest. Our results clearly show that neither a partition based on national
party nor on a partition based on member states is significant. To be clear,
this does not imply that MEPs of the same national party do not vote similarly
(because they do), rather, it means they vote highly similar to MEPs of other
parties.  For member states however, the division seems to run across member
states, and MEPs of the same member state do not necessarily vote in a similar
fashion. This shows that in general MEPs do not vote along national lines,
although for certain votes the national background may play a
role~\cite{Noury2002,Hix2009}.

The partition in EP groups shows 5 to 15 times the significance of a random
graph, making it quite significant. Whereas the partition into member states and
national parties remains almost constant throughout time, the partition into EP
groups increases quite a lot from 1979 to 2008, with an all time low of
$\frac{\mathcal{S}}{\langle \mathcal{S} \rangle} \approx 3.9$ in 1981 and
reaching its maximum of $\frac{\mathcal{S}}{\langle \mathcal{S} \rangle} \approx
16.8$ in 2001, an increase of more than $400\%$.  One possible explanation of
the general increase in divisiveness is that the EP has become more powerful
over the years, so that competition over important issues have taken a
lead~\cite{Noury2002,Kreppel2003,Hix2005}. Besides a general trend upwards,
there seems to be a particularly large jump between 1995 and 1996. One possible
explanation is that Austria, Finland and Sweden entered the European Union in
1995, whereafter MEPs were elected to parliament in 1995 and 1996.  On the other
hand, the accession of Eastern European countries in 2004 and Eastern Balkan
countries in 2007 did not seem to increase the divisiveness.  The maximum
significance closely follows the same trend as the EP group partition,
suggesting the two are related.

We have also analysed the sixth parliament for the year 2008, using CPM and
significance, to see what scales of community structure are present. We show
results for $\res=0.5$ and $\res=0.8$, with the latter corresponding to the
maximal significance for CPM. Clearly, the communities have a quite high
internal density, and are quite strongly connected amongst each other, as is
also clear from the adjacency matrices displayed in Fig.~\ref{fig:MEP}. 

At $\res=0.5$ CPM groups together the \emph{Greens/European Free Alliance}
(G/EFA) and the \emph{European United Left/Nordic Green Left} (EUL/NGL), which
are both left wing environmental parties. The \emph{Party of European
Socialists} (PES), joins the two other leftist parties at a somewhat higher
resolution of $\res=0.8$. The more conservative parties of the \emph{Union for
Europe of the Nations} (UEN) and the \emph{Alliance of Liberals and Democrats
for Europe} (ALDE) seem to join forces with the more centric \emph{European
People's Party-European Democrats} (EPP-ED). The eurosceptic
\emph{Independence/Democrats} (IND/DEM) group divides itself between the
right-wing and the left-wing bloc, although some members constitute a separate
bloc with other \emph{Non-Attached} (NA) MEPs, who themselves also split across
the two large blocs. The partition maximizing significance is different still,
but shows a similar grouping of EP groups, in addition to several smaller
communities. Surprisingly however, a part of UEN is joined with PES, although
they seem ideologically more remote.

These three different partitions highlight different aspects of the voting
network. The partition maximizing significance for CPM (at $\res=0.8$) seems to
highlight a more or less traditional partition into left and right wing
politics~\cite{Kreppel2003}. The partition for $\res=0.5$ seems to reveal a
grand coalition~\cite{Kreppel2003}, with mainly the green-left differing from
the rest. The partitions maximizing significance itself seems to highlight some
interesting split of the UEN. In conclusion, the EP shows signs of multiple
possible partitions, and significance seems to point to some interesting
partitions.

\section{Discussion}
\label{sec:Discussion}

We have presented in this paper a method to find significant scales in community
structure. Firstly, we introduced a bisectioning method allowing a fast and
accurate construction of a resolution profile. Secondly, we suggested a measure
based on subgraph probabilities in order to state what partitions are
significant. This measure can be interpreted as the gain in encoding a graph by
making use of a partition. We showed significance is able to accurately portray
partitions in benchmarks. Additionally, we showed on an empirical example using
voting data of the European Parliament that this measure conveys meaningful
information in that setting. Significance seems to be closely related to the
measure of surprise~\cite{Aldecoa2011,Aldecoa2013} and to stochastic block
models~\cite{Karrer2011}, relationships we hope to explore further in the
future.

We conjectured that the maximum significance $\langle S \rangle \sim n \log n$
for random graphs, which allows researchers to compare the observed significance
to the expected significance. It constitutes a first step towards fully fledged
hypothesis testing of the significance of partitions. Nonetheless, a proof of
this behaviour is lacking so far. Moreover, the standard error needs to be
estimated still, although simulations show it is relatively small. Furthermore,
the significance is currently based on Erd\"os-R\'enyi graphs, but it might be
more realistic to take the degree distribution into account~\cite{Newman2012}.
Significance is not only useful for partitions found using community detection,
but also for partitions based on other node
characteristics~\cite{Bianconi2009a}, such as school grades~\cite{Stehle2011},
gender~\cite{Palchykov2012}, or dormitories~\cite{Traud2011}, similar to what we
did for the European Parliament, and as such we deem it to be a valuable
contribution to analysing partitions in complex networks.

\section{Methods}
\subsection{Subgraph probabilities}
\label{sec:subgraph_prob_proof}

We write $G \in \G(n,p)$ for a random graph $G$ from $\G(n,p)$, such that each
edge has independent probability $p$ of being included in the graph, the usual
Erd\"os-R\'enyi (ER) graphs. We use $|G| := |V(G)| = n$ for
the number of nodes and $\|G\| := |E(G)| = m$ for the number of edges. We use $H
\subseteq G$ to denote the fact that $H$ is an \emph{induced} subgraph of $G$.
We write $\Pr(H \subseteq \G(n,p))$ for the probability that $H$ is an induced
subgraph of a $G \in \G(n,p)$. Let $S(n_c,m_c) = \{ G \mid |G| = n_c, \|G\| =
m_c\}$ denote the set of all graphs with $n_c=|G|$ vertices and $m_c=\|G\|$
edges. Furthermore, we slightly abuse notation and write $\Pr(S(n_c, m_c)
\subseteq \G(n,p))$ for the probability that a graph $G \in \G(n,p)$ contains
one of the graphs in $S(n_c,m_c)$, i.e.
\begin{equation*}
  \Pr(S(n_c, m_c) \subseteq \G(n,p)) = \Pr(\bigcup_{H \in S(n_c, m_c)} H \subseteq \G(n,p)).
\end{equation*}

Let us denote by $X$ the random variable that represents the number of
occurrences of a subgraph with $n_c$ vertices and $m_c$ edges in a random graph.
Let $X_H$ be the indicator value that specifies whether a subgraph $H$ of order
$n_c=|H|$ in the random graph equals one of the graphs in $S(n_c, m_c)$, which
of course comes down to
\begin{equation*}
  X_H = \begin{cases}
    1 & \text{if~} \|H\| = m_c \text{~and~} |H|=n_c \\
   0 & \text{otherwise}
        \end{cases}.
\end{equation*}
We can then write $X = \sum_H X_H$ where the sum runs over all $\binom{n}{n_c}$
possible subgraphs $H$. Obviously then, $\Pr(X > 0) = \Pr(S(n_c, m_c) \subseteq
\G(n,p))$. By Chauchy-Schwarz's inequality $\E(XY)^2 \leq E(X^2)E(Y^2)$ and
Markov's inequality $\Pr(X \geq a) \leq \frac{\E(X)}{a}$ we obtain the following
bounds
\begin{equation}
  \frac{\E(X)^2}{\E(X^2)} \leq \Pr(X > 0) \leq \E(X).
  \label{equ:markov_cauchy}
\end{equation}
This way of estimating probabilities is known as the second moment
method~\cite{Bollobas2001}.

It is convenient to define the probability that a graph of $n_c$ nodes contains
$m_c$ edges
\begin{align}
  r &= \Pr(S(n_c, m_c) \subseteq \G(n_c,p)) \nonumber \\
  &= \binom{\binom{n_c}{2}}{m_c} p^{m_c}(1 - p)^{\binom{n_c}{2} - m_c}.
  \label{equ:prob_simple}
\end{align}

\begin{theorem}
  The expected number of occurrences of an induced subgraph with $n_c$ nodes and
  $m_c$ edges in a random graph with $n$ nodes and density $p$, is given by
  \begin{equation}
    \E(X) = \binom{n}{n_c} r
  \end{equation}
  \label{thm:expected}
\end{theorem}
\begin{proof}
  By linearity of expectation, we have $\E(X) = \sum_H \E(X_H)$, and because
  $X_H$ is an indicator variable $\E(X_H) = \Pr(X_H = 1)$. Notice that $H$ has
  $n_c$ nodes, so that $H \in \G(n_c, p)$, and $\Pr(X_H = 1) = r$. There are
  $\binom{n}{n_c}$ subgraphs of $n_c$ nodes in a graph with $n$ nodes, which
  concludes the proof.
\end{proof}

For $\E(X^2)$ the idea is to calculate the expected value of the number of pairs
of subgraphs that have $m_c$ edges. We do this by separating in three parts: the
parts of the two subgraphs without overlap, and the part that overlaps.

\begin{theorem}
  The expected squared number of occurrences of an induced subgraph can be
  written as
  \begin{multline}
    \E(X^2) = \E(X) \sum_{u=0}^{n_c} \binom{n_c}{u}\binom{n - n_c}{n_c - u} 
    \\ \sum_{m(\Delta)}^{\min(\binom{u}{2},m_c)} \binom{M(u)}{m_c - m(\Delta)}
        \\ p^{m_c-m(\Delta)}(1 - p)^{M(u) - (m_c - m(\Delta))},
        \label{equ:EX^2}
  \end{multline}
  with $M(u)= \frac{n_c(n_c-1) - u(u-1)}{2}$.
\end{theorem}
\begin{proof}
The variable $X^2$ can be decomposed into parts $X_H \times X_{H'}$, such that
we need to investigate the probability that both $H$ and $H'$ have $m_c$
edges. So, we can separate this expectancy in parts of partially overlapping
subgraphs, like
\begin{equation}
  \E(X^2) = \sum_u \sum_{|H \cap H'|=u} \Pr(\|H\| = \|H'\| = m_c),
\end{equation}
where $u$ represents the overlap between the different subgraphs. If $H$ and
$H'$ are (edge) independent, so when $u < 1$, the answer is simply given by
$\Pr(X_H = 1)^2$. For $u \geq 1$ the answer is more involved.

So let us consider two subgraphs $H$ and $H'$ such that $|H \cap H'|=u\geq1$.
Let us separate this in three independent parts, the overlap $\Delta = H \cap
H'$, and the remainders $A = H - \Delta$ and $B = H' - \Delta$. Clearly then,
$|\Delta| = u$, and $|A| = |B| = n_c - u$. The probability that
$\|H\|=\|H'\|=m_c$ can then be decomposed in the probability that the sum of
these independent parts sum to $m_c$. The probability that $\|H\|=m_c$ can be
decomposed as
\begin{multline*}
  \Pr(\|H\|=m_c) = \sum_{m(\Delta)} \Pr(\|\Delta\| = m(\Delta)) \\
                              \Pr(\|H\| = m_c~\mid~\|\Delta\| = m(\Delta)).
\end{multline*}
where $m(\Delta)$ signifies the number of edges within $\Delta$. Similarly, we
arrive at the conditional probability for both subgraphs $H$ and $H'$.
However, since we have conditioned exactly on the overlapping part, the two
remaining parts are independent, and we can write
\begin{multline*}
  \Pr(\|H\| = \|H'\| = m_c~\mid~\|\Delta\| = m(\Delta)) = \\
        \Pr(\|H\| = m_c~\mid~\|\Delta\| = m(\Delta))^2.
\end{multline*}
This probability can be calculated and yields
\begin{multline*}
  \Pr(\|H\| = m_c~\mid~\|\Delta\| = m(\Delta)) = \\
  \binom{M(u)}{m_c - m(\Delta)}p^{m_c - m(\Delta)}(1 - p)^{M(u) - (m_c - m(\Delta))},
\end{multline*}
where $M(u)= \frac{n_c(n_c-1) - u(u-1)}{2}$. We then obtain
\begin{multline*}
 \Pr(\|H\| = \|H'\| = m_c) = \sum_{m(\Delta)} \Pr(\|\Delta\| = m(\Delta)) \\
 \binom{M(u)}{m_c - m(\Delta)}^2 p^{2(m_c - m(\Delta))}(1 - p)^{2M(u) - 2(m_c - m(\Delta))}
\end{multline*}
which leads to
\begin{multline*}
  \binom{\binom{n_c}{2}}{m_c} p^{m_c}(1 - p)^{\binom{n_c}{2} - m_c}
  \sum_{m(\Delta)} \binom{M(u)}{m_c - m(\Delta)} \\ 
 p^{m_c-m(\Delta)} (1 - p)^{M(u) - (m_c - m(\Delta))},
\end{multline*}
where $m(\Delta)$ ranges from $0$ to the minimum of $m_c$ and the number of
possible edges $\binom{u}{2}$.

Now counting the number of subgraphs that overlap in $u$ nodes, for each
choice of subgraph $H$, we choose $u$ nodes in $H$, and $n_c - u$ nodes in the
remaining $n - n_c$ nodes. In total, there are then
\begin{equation*}
  C_u = \binom{n}{n_c}\binom{n_c}{u}\binom{n - n_c}{n_c - u} 
\end{equation*}
overlapping subgraphs with $u$ nodes in common. Concluding, we arrive at
\begin{equation*}
  \E(X^2) = \sum_u C_u \Pr(\|H\| = \|H'\| = m_c~\mid~ | H \cap H' | = u).
\end{equation*}
Writing this out, we arrive at \eqtxt~\eqref{equ:EX^2}.  
\end{proof}

We consider subgraphs of size $s n$, with $0 < s < 1$ with fixed density $q$.
For the asymptotic analysis, we can afford to be a bit sloppy with this density,
and consider $(sn)^2$ possible edges in the subgraph of $sn$ nodes, so that $m_c
= q(sn)^2$, and we now denote by $S(s n,q)$ the subgraphs with density $q$
instead of the actual number of edges.

\begin{theorem}
  The probability for a dense subgraph can be bounded below and above
  asymptotically as 
  \begin{equation}
    \Pr(S(s n, q) \subseteq \G(n,p)) = e^{\Theta(-(sn)^2 D(q \midd p))}
  \end{equation}
  where $D(q \midd p)$ is the Kullback-Leibler divergence
  \begin{equation}
    D(q \midd p) = q \log \frac{q}{p} + (1 - q) \log \frac{1 - q}{1 -p}.
  \end{equation}
\end{theorem}
\begin{proof}
  We prove the asymptotic result by showing that both an upper and a lower bound
  have a similar asymptotic behaviour. The upper and lower bounds are provided
  by Markov's and Cauchy-Schwarz's inequality as stated in
  \eqtxt~\eqref{equ:markov_cauchy}. We will first prove the upper bound. Taking
  logarithms on Stirling's approximation, we obtain that
  \begin{equation*}
    \log \binom{n}{n_c} \sim n H\left(\frac{n_c}{n}\right) = n H(s),
  \end{equation*}
  where $H(p)$ is the binary entropy
  \begin{equation}
    H(p) = -p \log p - (1-p) \log (1 - p).
  \end{equation}
  We apply this to $\E(X) = \tbinom{n}{n_c} r$ with $r$ as in
  \eqtxt~\eqref{equ:prob_simple} and we obtain
  \begin{multline*}
    \log \E(X) \sim nH(s) + (sn)^2H(q) + \\
      \log \left( p^{q(sn)^2} (1-p)^{(1-q)(sn)^2} \right),
  \end{multline*}
  which can be simplified to $\log \E(X) \sim nH(s) - (sn)^2D(q \midd p)$,
  utilising the binary Kullback-Leibler divergence~\cite{Cover2012}
  \begin{equation}
    D(q \midd p) = q \log \frac{q}{p} + (1 - q) \log \frac{1-q}{1-p},
  \end{equation}
  which yields the upper bound by Markov's inequality.

  We need the second moment for the lower bound. This can be rewritten as
  $\E(X)^2 = \E(X) \sum_{u} \sum_{m(\Delta)} f(u, m(\Delta))$, with
  \begin{multline}
    f(u, m(\Delta)) = \binom{n_c}{u}\binom{n - n_c}{n_c - u} \binom{M(u)}{m_c - m(\Delta)} \\
         p^{m_c-m(\Delta)}(1 - p)^{M(u) - (m_c - m(\Delta))}.
  \end{multline}
  By Cauchy-Schwarz inequality, we want that
  \begin{equation*}
    \log \frac{\E(X)^2}{\E(X^2)} = \log \E(X) - \log \sum_{u} \sum_{m(\Delta)} f(u, m(\Delta))
  \end{equation*}
  increases as $-(sn)^2 D(q \midd p)$. We know that by Jensen's inequality we
  have
  \begin{equation}
    \log \frac{\E(X)^2}{\E(X^2)} \geq \log \E(X) - \sum_{u} \log \sum_{m(\Delta)} f(u, m(\Delta))
  \end{equation}
  Using the notation $u = \alpha sn$ we can write
  \begin{multline*}
    f(\alpha, m(\Delta)) = \binom{sn}{\alpha sn} \binom{(1 - s)n}{s(1 - \alpha) n}
    \binom{(1 - \alpha^2)(sn)^2}{m_c - m(\Delta)} \\
     p^{m_c - m(\Delta)} (1 - p)^{(1 - \alpha^2)(sn)^2 - (m_c - m(\Delta))}
  \end{multline*}
  We can bound
  \begin{multline*}
    \sum_{m(\Delta) = 0}^{(sn)^2 \min \alpha^2, q} \binom{(1 - \alpha^2)(sn)^2}{m_c - m(\Delta)} \\
     p^{m_c - m(\Delta)} (1 - p)^{(1 - \alpha^2)(sn)^2 - (m_c - m(\Delta))}
  \end{multline*}
  by
  \begin{equation*}
     \sum_{k = 0}^{q(sn)^2} \binom{(1 - \alpha^2)(sn)^2}{k} \\
     p^{k} (1 - p)^{(1 - \alpha^2)(sn)^2 - k},
  \end{equation*}
  with $k = m_c - m(\Delta)$,  in which we recognize the binomial cumulative
  probability $\Pr( Y \leq m_c )$ where $Y$ are the number of edges in the
  overlapping part. By Hoeffdings inequality this can be bounded by
  \begin{equation}
    \exp\left( - 2 (sn)^2 \frac{(1 - \alpha^2 - q)^2}{1 - \alpha^2} \right).
  \end{equation}
  Combining with our earlier result on $E(X)$, we then have
  \begin{multline*}
    \log \frac{\E(X)^2}{\E(X^2)} \geq -(sn)^2 D(q \midd p) + nH(s) - \\
    \sum_\alpha s n H(\alpha) + (1 - s)n H\left((1 - \alpha)\frac{s}{1 -s}\right) - \\
    \left(2 (sn)^2 \frac{(1 - \alpha^2 - q)^2}{1 - \alpha^2} \right).
  \end{multline*}
  For large enough $n$ the quadratic term dominates, and we obtain $\log
  \frac{\E(X)^2}{\E(X^2)} \geq -(sn)^2 D(q \midd p)$, giving the lower bound. By
  combining the lower and upper bound we obtain the asymptotic result stated in
  the theorem.
\end{proof}

\subsection{Optimizing significance}
\label{sec:optimize_sig}

As is common in the Louvain method~\cite{Blondel2008}, we look at the difference
of moving some node.  However, we also need to aggregate the graph, and still
correctly move communities. For that we need the node size $n_i$, similar as for
CPM~\cite{Traag2011}, which initially is $n_i=1$. Upon aggregating the graph the
node size is set to the sum of the node sizes within a community. Moving node
$i$ from community $r$ to $s$ with size $n_i$, $e_{ir}$ edges to community $r$
and $e_{is}$ edges to community $s$ gives a difference in significance of 
\begin{multline*}
  \Delta \Sig(\sigma) = 
  \binom{n_r}{2} D(q_r \midd p) - \binom{n_r - n_i}{2} D(q_r' \midd p) \\
  - \binom{n_s}{2} D(q_s \midd p) + \binom{n_s + n_i}{2} D(q_s' \midd p),
\end{multline*}
where $q_r' = \frac{m_r - e_{ir}}{ \binom{n_r - n_i}{2} }$
and
$q_s' = \frac{m_s + e_{is}}{ \binom{n_s + n_i}{2} }$. 

\bibliography{bibliography}

\begin{acknowledgments}
We acknowledge support from \emph{Actions de recherche concert\'ees, Large
Graphs and Networks} of the \emph{Communaut\'e Fran\c caise de Belgique} and
from the Belgian Network DYSCO (Dynamical Systems, Control, and Optimization),
funded by the Interuniversity Attraction Poles Programme, initiated by the
Belgian State, Science Policy Office.
\end{acknowledgments}

\end{document}